\documentclass[11pt]{article}

\usepackage{latexsym}
\usepackage{amsmath}
\usepackage{amssymb}
\usepackage{xspace}

\usepackage{algorithm}
\usepackage[noend]{algorithmic}
\usepackage{multicol}
\usepackage[all]{xy}
\usepackage{todonotes}

\usepackage{todonotes}

\usepackage{calc,ifthen}
\newcounter{hours}
\newcounter{minutes}
\newcommand{\Printtime}{\setcounter{hours}{\time/60}%
\setcounter{minutes}{\time-\value{hours}*60}%
\thehours:%
\ifthenelse{\value{minutes}<10}{0}{}\theminutes}

\newtheorem{xdefinition}{Definition}
\newtheorem{xobservation}{Observation}
\newtheorem{xtheorem}{Theorem}
\newtheorem{xlemma}{Lemma}
\newtheorem{xproposition}{Proposition}
\newtheorem{xcorollary}{Corollary}
{\hspace*{\fill}\raisebox{-1pt}{\boldmath$\Box$}\end{xdefinition}}
{\hspace*{\fill}\raisebox{-1pt}{\boldmath$\Box$}\end{xobservation}}
\newenvironment{theorem}{\begin{xtheorem}\rm}{\end{xtheorem}}
\newenvironment{lemma}{\begin{xlemma}\rm}{\end{xlemma}}
\newenvironment{proposition}{\begin{xproposition}\rm}{\end{xproposition}}
\newenvironment{corollary}{\begin{xcorollary}\rm}{\end{xcorollary}}
\newenvironment{proof}{\begin{trivlist}\item[]{\bf Proof }}%
{\hspace*{\fill}\raisebox{-1pt}{\boldmath$\Box$}\end{trivlist}}

\newcommand{\ALG}{\ensuremath{\operatorname{\textsc{A}}}\xspace}
\newcommand{\ALGB}{\ensuremath{\operatorname{\textsc{B}}}\xspace}
\newcommand{\OPT}{\ensuremath{\operatorname{\textsc{Opt}}}\xspace}
\newcommand{\DNF}{\ensuremath{\operatorname{\textsc{DNF}}}\xspace}
\newcommand{\nextfit}{\ensuremath{\operatorname{\textsc{Next-Fit}}}\xspace}

\newcommand{\harm}{{\ensuremath{\textsc{Harmonic}_k}}\xspace}

\newcommand{\Har}{{\ensuremath{\textsc{DHarmonic}_k}}\xspace}
\newcommand{\DHk}{{\ensuremath{\textsc{DH}_k}}\xspace}
\newcommand{\DHone}{{\ensuremath{\textsc{DH}_{1}}}\xspace}
\newcommand{\DHtwo}{{\ensuremath{\textsc{DH}_{2}}}\xspace}
\newcommand{\DHi}{{\ensuremath{\textsc{DH}_i}}\xspace}
\newcommand{\DHj}{{\ensuremath{\textsc{DH}_j}}\xspace}

\newcommand{\SET}[1]{\left\{#1\right\}}
\newcommand{\SETOF}[2]{\{#1 \mid #2\}}

\newcommand{\SEQ}[1]{\langle #1 \rangle}

\newcommand{\WR}{\ensuremath{\mathrm{WR}}\xspace} 
\newcommand{\CR}[1]{\ensuremath{\mathrm{CR}(#1)}\xspace}  
 
\newcommand{\CRab}[1]{\ensuremath{\mathrm{CR}_{a,b}(#1)}\xspace} 
\newcommand{\CRabu}{\ensuremath{\mathrm{CR}_{a,b}}\xspace} 
\newcommand{\RWOR}[2]{\ensuremath{\WR(#1,#2)}\xspace}
\newcommand{\PROB}[1]{\ensuremath{\operatorname{Prob}[#1]}}
\newcommand{\EXP}[1]{\ensuremath{\operatorname{E}[#1]}}
\newcommand{\EXPDIST}[2]{\ensuremath{\operatorname{E}_{#1}[#2]}}

\newcommand{\RO}[1]{\ensuremath{\mathrm{RR}(#1)}\xspace}

\newcommand{\MINV}[1]{\ensuremath{\mathrm{MR_{vol}}(#1)}}
\newcommand{\vol}[1]{\ensuremath{\textit{vol}(#1)}\xspace}

\newcommand{\FLOOR}[1]{\left\lfloor#1\right\rfloor}

\newcommand{\MINinline}[1]{\min\{#1\}}

\newcommand{\RDNF}{\ensuremath{R_{(0,\frac{1}{k})}}\xspace} 
\newcommand{\RHar}{\ensuremath{R_{[\frac{1}{k},1)}}\xspace}
\newcommand{\RDNFtwo}{\ensuremath{R_{(0,\frac{1}{2})}}\xspace} 
\newcommand{\RHartwo}{\ensuremath{R_{[\frac{1}{2},1)}}\xspace}
\newcommand{\gf}{\ensuremath{\psi_1}\xspace}
\newcommand{\ERU}[1]{\ensuremath{\mathrm{ER_U}(#1)}\xspace} 

\newcommand{\p}{p}
\newcommand{\bin}{\ensuremath{\beta}}
\newcommand{\eps}{\ensuremath{\varepsilon}\xspace}

\title{Online Bin Covering: Expectations vs.\ Guarantees\,\thanks{A
preliminary version of this paper appeared in the
proceedings of the Seventh Annual International Conference on
Combinatorial Optimization and Applications, 2013.
Supported in part by the Danish Council for Independent Research
and the Villum Foundation.}}

\author{Marie G. Christ \hspace{2em} Lene M. Favrholdt \hspace{2em} Kim S. Larsen \\[1ex]
        University of Southern Denmark \\
        Odense, Denmark \\[1ex]
        {\tt \{christm,lenem,kslarsen\}@imada.sdu.dk}}

\date{February 27, 2014}

\begin{document}

\maketitle

\begin{abstract}
Bin covering is a dual version of classic bin packing.
Thus, the goal is to cover as many bins as possible, where covering a bin
means packing items of total size at least one in the bin.

For online bin covering, competitive analysis fails to distinguish between
most algorithms of interest; 
all ``reasonable'' algorithms have a competitive ratio of $\frac12$.
Thus,
 in order to get a better understanding of the combinatorial difficulties in solving this problem,
 we turn to other
performance measures, namely relative worst order, random order, and
max/max analysis, as well as
analyzing input with restricted or uniformly distributed item sizes.
In this way, our study also supplements the ongoing systematic studies of
the relative strengths of various performance measures.

Two classic algorithms for online bin packing that have natural dual
 versions are \harm and \nextfit. 
Even though the algorithms are quite different in nature, 
the dual versions are
 not separated by competitive analysis.
%
We make the case that when guarantees are needed, even under restricted
input sequences, dual \harm is preferable. In addition, we
establish quite robust theoretical results showing that if items come
from a uniform distribution or even if just the ordering of items is uniformly
random, then dual \nextfit is the right choice.
\end{abstract}

\section{Introduction}
Bin covering~\cite{AJKL84j} is a dual version of classic bin packing.
As usual, bins have
size one and items with sizes between zero and one must be packed. However,
in bin covering, the objective is to cover as many bins as possible,
where a bin is covered if the sizes of items placed in the bin sum up
to at least one.
We are considering the online version of bin covering. A problem
is online if the input sequence is presented to the algorithm one item
at a time, and the algorithm must make an irrevocable decision regarding
the current item without knowledge of future items.

Bin covering algorithms have numerous important applications.
For instance, when packing or canning food items guaranteeing a
minimum weight or volume, reductions in the overpacking of even a
few percent may have a large economic impact.
If items arrive on a conveyor belt, for instance, the problem becomes online.

Classic algorithms for online bin packing are \nextfit and the
parameterized family \harm~\cite{HL85}.
\nextfit is a very simple and natural algorithm, and \harm was designed to
obtain a competitive ratio~\cite{ST85,KMRS88} better than any Any-Fit
algorithm (First-Fit and Best-Fit are examples of Any-Fit
algorithms for bin packing, and the competitive ratio of Next-Fit is
worse than both these algorithms).
\harm and variations of it have been analyzed extensively~\cite{RBLL89,W93,S02}.
We consider the obvious dual version of these,
\DNF~\cite{AJKL84j} and \DHk~\cite{CW98}.
These algorithms are quite different in nature
and the bin packing versions are clearly separated, 
having competitive ratios of $2$ and approximately $1.691$, respectively.
However, for bin covering,
competitive analysis does not distinguish between them!
In fact, for bin covering, competitive analysis categorizes both
algorithms as being optimal among deterministic algorithms, but also worst
possible among ``reasonable'' algorithms for the problem.
This is unlike the situation in bin packing, and in general, results
from bin packing do not transfer directly to bin covering.

To understand the algorithmic differences better,
it is therefore necessary to employ different techniques,
and we turn to other generally applicable performance measures, namely
relative worst order analysis, random order analysis, and
max/max analysis.
As for almost all performance measures, the idea is to abstract away
some details of the problem to enable comparisons. Without some abstraction,
it is hard to ever, analytically, claim that one algorithm
is better than another, since almost any algorithm performs better
than any other algorithm on at least one input sequence.
For all the measures considered here, the abstraction can be viewed
as being defined via first a partitioning of the set of input sequences of a given length and then
an aggregation of the results from each partition.
For each sequence length, competitive analysis, for instance,
considers all the ratios of the online performance to the optimal
offline performance obtained for each sequence of that length,
and then takes the worst ratio of all of these.
The measures above employ a less fine-grained partitioning of the
input space.
Worst order and random order analysis group permutations of the same
sequence together instead of considering each sequence separately,
deriving worst-case or average-case performance, respectively, within each partition.
With max/max analysis the partitioning of the input space is even
coarser: for each sequence length $n$, the online worst-case behavior
over all sequences of length $n$ is compared to the worst-case optimal
offline behavior over all sequences of length $n$.
There is no one correct way to compare algorithms,
but since these measures focus on different aspects of algorithmic behavior,
considering all of the ones above lead to a very broad analysis of the problem.
Extensive motivational sections can be found in the papers introducing
these measures and in the survey~\cite{DL05}.
As a further supplement, we analyze restricted input sequences, where
items have similar size, which is likely to happen in practice if one
is packing products with an origin in nature, for instance.
Finally, we consider input sequences containing items having uniformly distributed sizes.

Relative worst order analysis~\cite{BF07j,BFL07j} has been applied to
many problems; a recent list can be found
in~\cite{EKL13j}.
In~\cite{EFK12}, bin covering was analyzed, but using
a version of the problem allowing items of size~1.
We analyze the more commonly studied version for bin covering,
where all items are
strictly smaller than 1. Since worst-case sequences from~\cite{EFK12}
contain items of size~1, this leads to slightly different results.
For completeness, we include these results.
Random order analysis~\cite{K96}
was introduced for classic bin packing, but has also been used
for other problems; a server problem, for instance~\cite{BIL09p}.
Max/max analysis~\cite{BB94} was introduced
as an early step towards refining the results from competitive
analysis for paging and a server problem.


Relative worst order analysis emphasizes the fact that there
exist multisets of input items where \DNF can perform
$\frac{3}{2}$ times as poorly as \DHk. On the other hand,
\DHk's method of limiting the worst-case also means that it has less
of an opportunity to reach the best case, as opposed to \DNF.
This is reflected in the random order analysis, where \DNF
comes out at least as well as \DHk.
Another way of approaching randomness is to analyze a uniform distribution.
We establish new results on \DHk showing that its performance here
is slightly worse than that of \DNF, in line with the random order results.
With the max/max analysis, a distinction between the two algorithms 
can only be achieved, when the item sizes are limited, and
\DHk is the algorithm selected as best by this measure.
With respect to competitive analysis,
we also consider restricted input in the sense that item sizes may
only vary across one or two consecutive \DHk partitioning points. This is a
formal way of treating the case where items are of similar size,
while allowing greater variation when this size is large. We show
that with this restricted form of input, considering the
worst-case measures of competitive analysis, \DHk is deemed
better than \DNF, as \DNF is more vulnerable to
worst-case sequences. 

This study also contributes to the ongoing systematic studies of
the relative strengths of various performance measures,
initiated in~\cite{BIL09p}. Up until that paper, most performance
measures were introduced for a specific problem to 
overcome the limitations of
competitive analysis. In~\cite{BIL09p}, comparisons
of performance measures different from competitive analysis
were initiated, and this line of work has been continued
in~\cite{BGL12p,BLM12p,BGL13}, among others.
Our results supplement results in~\cite{C88},
showing that
no deterministic algorithm for the bin covering problem
can be better than $\frac{1}{2}$-competitive and
giving an asymptotically optimal algorithm for the case of items
being uniformly distributed on $(0,1)$.
For \DNF, \cite{CFGK91} established an expected competitive
ratio of $\frac{2}{e}$ under the same conditions.

In the following, we formally define the bin covering problem and the
algorithms \DNF and \DHk, the performance of which we compare under
different performance measures. The performance measures themselves are
defined in each their section.
We conclude on our findings in the final section.

\subsection*{Bin Covering}
In the one dimensional bin covering problem, the algorithm gets an input
sequence $I=\SEQ{i_1,i_2,\dots}$ of item sizes, where for all $j$, $0 < i_j < 1$.
The items are to be packed in bins of size 1.
A bin is {\em covered}, if items of total size at least 1 have been
 packed in it, and
the goal is to cover as many bins as possible.

Requiring items to be strictly smaller than 1 corresponds to assuming
 that items of size 1 are treated separately.
This makes sense, since there is no advantage in combining an item of
 size 1 with any other items in a bin.
In other words, any algorithm not giving special treatment to items
of size 1 could trivially be improved by doing so.

For a bin covering algorithm \ALG, we let $\ALG(I)$ denote the number
 of covered bins when given the sequence $I$ of items. 
We let $\OPT$ denote an optimal offline algorithm.
Thus, $\OPT(I)$ is the largest number of bins that can be covered
by any algorithm processing $I$.

In algorithms for bin packing and covering, it is standard to use the
following terminology.
A bin that has received at least one item is \emph{open} if it may
receive more items, and \emph{closed}
if the algorithm will not consider
that bin again for future items.

\subsubsection*{The Dual Next-Fit algorithm}
Assmann, Johnson, Kleitman, and Leung~\cite{AJKL84j} introduced the
Dual \nextfit algorithm 
(\DNF), an adaptation of the \nextfit algorithm for bin packing.
\DNF always keeps at most one open bin. 
When a new item arrives, it is packed in the currently open bin, if
 any.
Otherwise, a new bin is opened.
A bin is closed when it has received items of total size at least one.

\subsubsection*{The Dual Harmonic algorithm}
The algorithm \harm was introduced for bin packing by Lee and Lee~\cite{HL85}. This algorithm
partitions the interval $(0,1)$ into $k$ subintervals,
with the partitioning points at $\frac{1}{2},\frac{1}{3},\dots, \frac{1}{k}$,
resulting in the different sized intervals
$(0,\frac{1}{k}],(\frac{1}{k},\frac{1}{k-1}],\dots,(\frac{1}{2},1)$.
\harm packs items from each of
these $k$ subintervals in separate bins.
This means that each closed bin for the interval~$(\frac{1}{j},\frac{1}{j-1}]$ 
contains exactly $j$ items.
The natural adaptation to the bin covering problem is to use the intervals
$$\left(0,\frac{1}{k}\right),\left[\frac{1}{k},\frac{1}{k-1}\right),\dots,\left[\frac{1}{2},1\right)\,.$$
The resulting algorithm, \Har (\DHk), uses exactly $j$ items from the
interval~$[\frac{1}{j},\frac{1}{j-1})$ to cover a bin.
All through the paper we assume that $k\geq 2$, since for $k=1$,
\DHk becomes \DNF.

\section{Competitive Analysis}
\label{competitive-analysis}
In competitive analysis~\cite{ST85,KMRS88}, the performance of an online
algorithm is compared to that of an optimal offline algorithm \OPT. 
An algorithm \ALG for a maximization problem is called
\emph{$c$-competitive} if there exists a fixed 
constant~$b$ such that for any input sequence $I$,
it holds that $\ALG(I)\geq c\OPT(I) + b$.
The supremum over all such $c$ is the \emph{competitive ratio} \CR{\ALG} of \ALG.
Note that some authors reverse the order of the algorithm and \OPT to get
ratios larger than one.

For bin covering, Csirik and Totik~\cite{C88} showed that no deterministic
online algorithm can be better than $\frac{1}{2}$-competitive.
\DNF was shown to be $\frac{1}{2}$-competitive in~\cite{AJKL84j},
and the same result for \DHk was noted in~\cite{EFK12}.
For completeness, to show that this result is tight for a
large class of algorithms,
we define a \emph{reasonable} algorithm to be one that closes bins as
soon as they are covered, does not close bins before
they are covered, and does not have more than a constant number of open
bins at any point.
\begin{theorem}
\label{competitive-reasonable}
Any deterministic reasonable algorithm has a competitive ratio of
$\frac{1}{2}$.
\end{theorem}
\begin{proof}
The upper bound follows from~\cite{C88}.
For the lower bound, note that
the only item that can overfill a bin is the last item to go into that
bin, by the definition of a reasonable algorithm. Since that item has
size less than one, all bins will contain items of total size
less than two. Thus, \OPT could not cover more than twice as many bins,
using items from the closed bins.
Being reasonable also means that there are only a constant number of
open bins, so the items in there can only enable \OPT to cover
an additive constant of further bins. Thus,
no reasonable algorithm can be worse than $\frac{1}{2}$-competitive.
\end{proof}

\subsection{Limiting the item sizes}
In some applications of the bin covering problem it is likely that
the sizes of the items contained in an input sequence differ only
slightly, e.g., packing similar food items into a container, guaranteeing
the consumer a minimum weight.
In the following, we investigate the performance of \DNF and
\DHk on sequences with similar-sized items.
Since it seems reasonable to allow larger variance in size when the
considered sizes are large,
we consider sequences containing item sizes from two or three
consecutive \DHk intervals. 

We first consider intervals $(a,b)\subseteq
(0,1)$ that contain exactly one \DHk partitioning point.
Afterwards, we consider sequences with exactly two \DHk partitioning points.
We emphasize that there are no restrictions on the endpoints $a$ and $b$,
which can be any real numbers,
as long as the interval between them contains exactly one or two \DHk partitioning points.
In both cases, \DHk turns out to have the better ratio.

\begin{proposition}
\label{proposition-basic-dnf}
For any $x\in\mathbb{N}$, $x \geq 2$, and $\varepsilon > 0$,
$\CRab{\DNF} \leq \frac{x}{x+1}$, even if we
only consider items in the range
$[\frac{1}{x}-\eps,\frac{1}{x}+\eps]$.
\end{proposition}
\begin{proof}
Consider the sequence 
$\SEQ{\SEQ{\frac{1}{x}}^{x-1}, \frac{1}{x} - \eps, \frac{1}{x} + \eps}^{xn}$.
For this sequence, \DNF covers only $xn$ bins, whereas \OPT can
place exactly one small and one large item in each bin, filling up with
items of size $\frac{1}{x}$, to cover $(x+1)n$ bins.
\end{proof}

For any $(a,b) \subseteq (0,1)$, we let
\CRabu denote the competitive ratio on
sequences where all item sizes are in $(a,b)$.

If $(a,b)$ does not contain at least one of the interval borders
used by \DHk, then \DHk packs exactly like \DNF.
If $(a,b)$ contains a \DHk border, then we define 
 $$\frac{1}{\p} = 
 \max\left\{\frac{1}{l} \left\vert l \in \mathbb{N}, \frac{1}{l}
 < b\right.\right\},$$
and refer to $\frac{1}{\p}$ as the {\em maximal border in $(a,b)$}.

Note that if $(a,b)$ contains exactly one of the interval borders
used by \DHk, then $\frac{1}{\p+1}\leq a<\frac{1}{\p}$.
The next two theorems and the corollary deal with this case.

\begin{theorem}\label{lemma-DNF-2int}
If $\frac{1}{\p+1}\leq a<\frac{1}{\p}$, then 
\[\CRab{\DNF} = \frac{\p}{\p+1}\]
\end{theorem}
\begin{proof}
The lower bound follows directly from the fact that it takes at 
least~$\p$ and at most $\p+1$ items to cover a bin,
and the upper bound follows from Proposition~\ref{proposition-basic-dnf}.
\end{proof}

\begin{theorem}\label{lemma-Har-2int}
If $\frac{1}{\p+1} \leq a < \frac{1}{\p}$ and $k \geq \p$, then 
\[\CRab{\DHk} = \frac{\p^2+1}{\p(\p+1)}\]
\end{theorem}
\begin{proof}
We consider the lower bound first.
Since $k \geq \p$, \DHk packs the items of size larger than or equal
 to $\frac{1}{\p}$ in separate bins.
For any sequence $I$, let $t$ denote the total number of items in $I$ and
let $\ell$ denote the number of items of size larger than or equal to
$\frac{1}{\p}$.
Then, \DHk covers at least
$\lfloor\frac{\ell}{\p}\rfloor+\lfloor\frac{t-\ell}{\p+1}\rfloor>\frac{\ell+t\p}{\p(\p+1)}-2$
bins. 
Thus, letting $n=\OPT(I)$, we obtain 
$$\CRab{\DHk} \geq \frac{\ell+t\p}{n\p(\p+1)}.$$

We treat this in two cases:

\textbf{Case $\ell < n$:} At least $n-\ell$ bins covered by \OPT contain more than $\p$~items.
Hence, $t\geq n\p + n - \ell$,
and
\begin{align*}
\CRab{\DHk}
& \geq \frac{\ell+t\p}{n\p(\p+1)}\\
& \geq \frac{\ell+n \p^2+n\p-\ell\p}{n \p (\p+1)}\\
& =    \frac{n\p^2 + n + n(\p-1) - \ell(\p-1)}{n \p (\p+1)}\\
& >    \frac{n\p^2+n}{n \p (\p+1)}, \text{ since } \ell<n\\
& = \frac{\p^2+1}{\p(\p+1)}.
\end{align*}

\textbf{Case $\ell \geq n$:}
Here we can only use $t \geq n \p$, obtaining
\[
\CRab{\DHk}
\geq
\frac{\ell+t \p}{n \p (\p+1)} \geq \frac{n+t \p}{n \p (\p+1)} \geq \frac{n(1+\p^2)}{n\p(\p+1)}=\frac{\p^2+1}{\p(\p+1)}.
\]

For the upper bound,
we consider the sequence $\SEQ{\SEQ{\frac{1}{\p}- \frac{\varepsilon}{\p-1}}^{\p-1},\frac{1}{\p}+\varepsilon}^n$,
where $0 < \varepsilon < \MINinline{(\p-1)(\frac{1}{\p}-a),b-\frac{1}{\p}}$,
ensuring that both item sizes belong to $(a,b)$.
\OPT covers $n$ bins, whereas
\DHk packs the different sized items in separate bins,
and covers $\frac{n(\p-1)}{\p+1}+\frac{n}{\p} = \frac{\p^2+1}{\p(\p+1)}n$ bins,
up to an additive constant independent of $n$ which is due to rounding.
\end{proof}

It follows that if $(a,b)$ contains exactly one \DHk partitioning
point, $\frac{1}{\p}$, and $k \geq \p$, then \DHk has a better competitive ratio than \DNF:

\begin{corollary}
\label{corollary-2int}
 If $\frac{1}{\p+1}\leq a < \frac{1}{\p}$ and $k \geq \p$, then $$\CRab{\DHk} > \CRab{\DNF}\,.$$
\end{corollary}
\begin{proof}
 The result follows from Theorems~\ref{lemma-DNF-2int} and~\ref{lemma-Har-2int},
since $\CRab{\DHk} = \frac{\p^2+1}{\p(\p+1)} = \frac{\p}{\p+1} + \frac{1}{\p(\p+1)} = \CRab{\DNF} + \frac{1}{\p(\p+1)}$.
\end{proof}

We now consider intervals $(a,b)\subseteq
(0,1)$ that contain exactly two \DHk partitioning points,
 $\frac{1}{\p}$ and $\frac{1}{\p+1}$.
Including the extra partitioning point, $\frac{1}{\p+1}$, results in a
 lower competitive ratio for \DNF, with an upper bound depending on whether
 $b$ is smaller or larger than $\frac{\p+2}{\p(\p+1)}$.
The competitive ratio of \DHk becomes lower than with just one
partitioning point, only if $b > \frac{\p+2}{\p(\p+1)}$.


\begin{theorem}\label{lemma-DNF-3int-small}
If $a < \frac{1}{\p + 1} $, then
\[
\CRab{\DNF} \leq
\left\{
\begin{array}{ll}
  \displaystyle
  \frac{\p+1}{\p+2},          & 
  \displaystyle
  \mbox{if $b \leq \frac{\p+2}{\p(\p+1)}$} \\[2ex]
  \displaystyle
  \frac{\p^2+\p}{\p^2+2\p+2}, & 
  \displaystyle
  \mbox{otherwise}
\end{array}
\right.
\]
\end{theorem}
\begin{proof}
Applying Proposition~\ref{proposition-basic-dnf}, using $x=\p+1$,
we get an upper bound of $\frac{\p+1}{\p+2}$.

If $b > \frac{\p+2}{\p(\p+1)}$, we can strengthen the upper bound further.
Since
\[a < \frac{1}{\p+1} < \frac{1}{\p} < \frac{\p+2}{\p(\p+1)} < b < \frac{1}{\p-1},\]
we can choose an $\eps>0$ small enough so that for any fraction $r$ in this
sequence of inequalities, and any constant, $c$, used below,
$r-c\eps$ and $r+c\eps$ respect the same inequalities as $r$.
Now, we consider a sequence $I$ consisting of the following subsequences,
for some integer $n$:
 \begin{itemize}
  \item $\displaystyle \SEQ{ \SEQ{\frac{1}{\p}}^{\p-1},\frac{1}{\p}-2\varepsilon,\frac{\p+2}{\p(\p+1)}+\varepsilon }^{(\p+1)(\p-2)n}$ 
  \item $\displaystyle \SEQ{ \SEQ{\frac{1}{\p+1}}^{\p},\frac{1}{\p+1}-\varepsilon,\frac{\p+2}{\p(\p+1)}+\varepsilon}^{(p+1)n}$
  \item $\displaystyle \SEQ{ \frac{1}{\p+1} + i(\p-2)\varepsilon, \frac{1}{\p+1} - (i+1)(\p-2)\varepsilon, \SEQ{\frac{1}{\p+1}+\varepsilon}^{\p-2},
 \frac{1}{\p+1}-\varepsilon,\frac{\p+2}{\p(\p+1)}+\varepsilon}$
 for $i=1,2,\dots,(\p+1)n-1$
  \item $\displaystyle \SEQ{ \frac{1}{\p+1} - (\p-2)\varepsilon, \SEQ{\frac{1}{\p+1}+\varepsilon}^{\p-2},\frac{1}{\p+1}-\varepsilon,\frac{\p+2}{\p(\p+1)}+\varepsilon}$
  \item $\displaystyle \SEQ{ \frac{1}{\p+1}+ (\p+1)n(\p-2)\varepsilon}$
 \end{itemize}
Giving the items in this order, the number of bins covered by \DNF is
$$\DNF(I) = (\p+1)(\p-2)n+(\p+1)n+((\p+1)n-1)+1=\p(\p+1)n\,.$$

\OPT just puts the items into the correct bins
as they arrive, but for verification purposes, we list an order of the
items which is optimal for \DNF, but emphasize that this is \OPT obtaining
this result.
Below, we use that $\frac{1}{\p+1}+\frac{\p+2}{\p(\p+1)}=\frac{2}{\p}$.
The following order illustrates the optimal packing:
\begin{itemize}
 \item $\displaystyle \SEQ{ \frac{\p+2}{\p(\p+1)}+\varepsilon, \frac{1}{\p}-2\varepsilon, \SEQ{ \frac{1}{\p}}^{\p-3},\frac{1}{\p+1}+\varepsilon }^{(\p+1)(\p-2)n}$
 \item $\displaystyle \SEQ{ \frac{\p+2}{\p(\p+1)}+\varepsilon, \SEQ{\frac{1}{\p}}^{\p-2},\frac{1}{\p+1}-\eps }^{2(\p+1)n}$
 \item $\displaystyle \SEQ{ \frac{1}{\p+1} + i(\p-2)\varepsilon, \frac{1}{\p+1} - i(\p-2)\varepsilon, \SEQ{\frac{1}{\p+1}}^{\p-1}}$

 for $i=1,2,\dots,(\p+1)n$
 \item $\displaystyle \SEQ{\SEQ{\frac{1}{\p+1}}^{\p+1}}^{n}$
\end{itemize}
The number of bins covered by \OPT is
 $$\OPT(I) = (\p+1)(\p-2)n+2(\p+1)n+(\p+1)n+n=(\p^2+2\p+2)n\,.$$
\end{proof}

\begin{theorem}\label{lemma-Har-3int}
If $\frac{1}{\p+2}\leq a < \frac{1}{\p+1}$ and $k \geq \p+1$, then
\[
\CRab{\DHk} = 
\left\{
\begin{array}{ll}
\displaystyle 
\frac{\p^3+2\p^2+\p+2}{\p(\p+1)(\p+2)} = \frac{\p^2+1}{\p(\p+1)}, &
               \displaystyle \mbox{if $b \leq \frac{\p+2}{\p(\p+1)}$} \\[2ex]
\displaystyle 
\frac{\p^3+2\p^2+2}{\p(\p+1)(\p+2)},    &
               \displaystyle \mbox{otherwise}
\end{array}
\right.
\]
\end{theorem}
\begin{proof}
 We prove the lower bound first.

Items of size less than $\frac{1}{\p+1}$ are called {\em
  small}, items of size 
at least $\frac{1}{\p}$ are called {\em large}, and the remaining items
are called {\em medium}.
Let $s$, $m$, and $\ell$ denote the number of small, medium, and large
 items, respectively.

Consider an optimal packing.
For $i\in\SET{1,2,3}$,
let $n_i$ denote the number of bins with exactly $\p+i-1$ items.
Then,
$n = n_1 + n_2 + n_3$ is the number of bins covered by \OPT.
Since \DHk covers exactly $\lfloor \frac{s}{\p+2} \rfloor + \lfloor
\frac{m}{\p+1} \rfloor + \lfloor \frac{\ell}{\p} \rfloor$ bins,
independent of the order of the items,  we can
consider items from the three types of bins separately.

{\bf Bins with $\p$ items:}
Let $s_1$, $m_1$, and $\ell_1$ denote the number of small, medium, and
 large items, respectively, packed in these $n_1$ bins by \OPT.
Further, let $t_1 = s_1 + m_1 + \ell_1 = \p n_1$ denote the total
 number of items packed here.

In each of these bins, the items have an average size of at least
$\frac{1}{\p}$. This means that $\ell_1 \geq n_1$, since each bin has to contain at least one item of size at least $\frac{1}{\p}$. 

We now prove the inequality $s_1 \leq \ell_1$.
We do this by proving the stronger result that for each bin $\bin$ with
exactly $\p$ items, $s(\bin) \leq \ell(\bin)$, where $s(\bin)$ and $\ell(\bin)$
denote the number of small and large items in $\bin$, respectively.

Small items deviate from the average size with strictly more than
\begin{align*}
\varepsilon_s
  & = \frac{1}{\p} - \frac{1}{\p+1}
    = \frac{1}{\p(\p+1)}
    = \frac{\p-1}{(\p-1)\p(\p+1)}\,,
\end{align*}
and large items deviate with at most
\begin{align*}
\varepsilon_{\ell}
& = \frac{1}{\p-1} - \frac{1}{\p}
   = \frac{1}{(\p-1)\p}
   = \frac{\p+1}{(\p-1)\p(\p+1)}\,.
\end{align*}
Thus, having an average item size of at least $\frac{1}{\p}$ within a
  bin $\bin$ requires $\ell(\bin) \eps_{\ell} > s(\bin) \eps_s$.
Assume that $\bin$ contains more small items than large items, i.e.,
  $s(\bin) \geq \ell(\bin)+1$.
Then, $\ell(\bin) \eps_{\ell} > (\ell(\bin)+1) \eps_s$,
which is equivalent to
$\ell(\bin) > \frac{\eps_s}{\eps_{\ell}-\eps_s}$,
implying that $\ell(\bin) >   \frac{\p-1}{2}$,
using the equation above.
Since $\ell(\bin)$ is an integer, this means that $\ell(\bin) \geq
  \frac{\p}{2}$, and
since $\bin$ contains exactly $p$ items, this proves that $\ell(\bin) \geq
  s(\bin)$.

The contribution to the number of bins covered by \DHk from the $t_1$
items considered here is more than $d_1 -3$, where the $-3$ comes from
a possible fractional part in the three addends below.
\begin{align*}
d_1
 & = \frac{s_1}{\p+2} + \frac{m_1}{\p+1} + \frac{\ell_1}{\p}\\
 & = \frac{s_1}{\p+2} + \frac{\p n_1 - s_1 - \ell_1}{\p+1} + \frac{\ell_1}{\p}\\
 & = \frac{(\p+1)s_1}{(\p+1)(p+2)} + \frac{\p n_1}{\p+1} -
     \frac{(\p+2)s_1}{(\p+1)(\p+2)} - \frac{\p \ell_1}{\p(\p+1)} +
     \frac{(p+1)\ell_1}{\p(\p+1)} \\ 
 & = \frac{\p n_1}{\p+1} - \frac{s_1}{(\p+1)(\p+2)} + \frac{\ell_1}{\p(\p+1)} \\
 & \geq \frac{\p n_1 }{\p+1} - \frac{\p \ell_1}{\p(\p+1)(\p+2)}
      + \frac{(\p+2)\ell_1}{\p(\p+1)(\p+2)}, \text{ since } s_1 \leq \ell_1 \\
 & = \frac{\p^2 (\p+2) n_1 }{\p(\p+1)(\p+2)} +
     \frac{2\ell_1}{\p(\p+1)(\p+2)}\\
 & \geq  \frac{(\p^3+2\p^2)n_1}{\p(\p+1)(\p+2)} +
     \frac{2n_1}{\p(\p+1)(\p+2)}, \text{ since } \ell_1 \geq n_1\\
 & =  \frac{\p^3+2\p^2+2}{\p(\p+1)(\p+2)} n_1
\end{align*}

If $b \leq \frac{\p+2}{\p(\p+1)} = \frac{1}{\p} + \frac{1}{\p(\p+1)}$, 
 one large item is not large enough to compensate for the loss of contribution
to the average that a small item generates (recall that this loss is
strictly larger than $\varepsilon_s = \frac{1}{\p(\p+1)}$). Therefore,
additional to the $n_1$ large 
items, there has to be at least one more large item for each small
item, i.e., $\ell_1 \geq s_1 + n_1$.
In this case, we can strengthen the calculations above from a certain point:
\begin{align*}
 d_1
 & = \frac{\p n_1}{\p+1} - \frac{s_1}{(\p+1)(\p+2)} + \frac{\ell_1}{\p(\p+1)} \\
  & \geq \frac{\p n_1}{\p+1}-\frac{s_1}{(\p+1)(\p+2)}+\frac{s_1 + n_1}{\p(\p+1)}, \text{ since } \ell_1 \geq s_1 + n_1 \\
  & = \frac{(\p^2+1)n_1}{\p(\p+1)} +\frac{2s_1}{\p(\p+1)(\p+2)}\\
  & \geq \frac{\p^2+1}{\p(\p+1)}n_1, \text{ since } s_1 \geq 0 \\
  & = \frac{\p^3+2\p^2+\p+2}{\p(\p+1)(\p+2)}n_1
\end{align*}

{\bf Bins with $\p+1$ items:}
Let $s_2$, $m_2$, and $\ell_2$ denote the number of small, medium, and
 large items, respectively, packed in these $n_2$ bins.
Further, let $t_2 = s_2 + m_2 + \ell_2 = (\p+1) n_2$ denote the total
 number of items packed here.
 
 In each of these bins, the items have an average size of at least
$\frac{1}{\p+1}$.
This means that $s_2 \leq \p n_2 $, as each bin has to contain at least one item of size at least $\frac{1}{\p+1}$. 

The contribution to the number of bins covered by \DHk from the $t_2$
items considered here is more than $d_2 -3$, where
\begin{align*}
  d_2
  & = \frac{s_2}{\p+2} + \frac{m_2}{\p+1} + \frac{\ell_2}{\p}\\
  & = \frac{s_2}{\p+2} + \frac{(\p+1)n_2-s_2-\ell_2}{\p+1} +
      \frac{\ell_2}{\p}\\
  & = \frac{(\p+1)s_2}{(\p+1)(\p+2)} + n_2 -
      \frac{(\p+2)s_2}{(\p+1)(\p+2)} - \frac{\p \ell_2}{\p(\p+1)} + \frac{(\p+1)\ell_2}{\p(\p+1)}\\
  & = n_2 - \frac{s_2}{(\p+1)(\p+2)} + \frac{\ell_2}{\p(\p+1)}\\
  & \geq n_2 - \frac{\p n_2}{(\p+1)(\p+2)} + \frac{\ell_2}{\p(\p+1)}, \text{ since } s_2 \leq \p n_2 \\
  & \geq \frac{(\p^2+2\p+2)n_2}{(\p+1)(\p+2)}, \text{ since } \ell_2 \geq 0 \\
  & \geq \frac{\p^3+2\p^2+2\p}{\p(\p+1)(\p+2)}n_2 \\
  & \geq \frac{\p^3+2\p^2+\p+2}{\p(\p+1)(\p+2)}n_2, \text{ since } \p \geq 2 \\
\end{align*}

{\bf Bins with $\p+2$ items:}
Since \DHk cannot be forced to pack more than $\p+2$ items in each bin,
the contribution to the number of bins covered by \DHk from the
items considered here is exactly $n_3$.

Now, we turn to the upper bound.
Assume first that $b > \frac{\p+2}{\p(\p+1)}$.
Consider the sequence
$\SEQ{\SEQ{\frac{1}{\p+1}-\varepsilon}^n,
      \SEQ{\frac{\p+2}{\p(\p+1)}+(\p-1)\varepsilon}^n,
      \SEQ{\frac{1}{\p}-\varepsilon}^{n(\p-2)}}$
for some $\varepsilon>0$,
sufficiently small such that all the items in the sequence are in
the range $(a,b)$.
Since $\frac{1}{\p+1} + \frac{\p+2}{\p(\p+1)} = \frac{2}{\p}$,
\OPT can cover $n$ bins by combining one
item of size $\frac{1}{\p+1}-\varepsilon$, one item of size 
$\frac{\p+2}{\p(\p+1)}+(\p-1)\varepsilon$, and $(\p-2)$ items of size
$\frac{1}{\p}-\varepsilon$. 
\DHk packs each kind of item
separately, covering $\frac{n}{\p+2}+\frac{n}{\p}+\frac{n(\p-2)}{\p+1} =
\frac{\p^3+2\p^2+2}{\p(\p+1)(\p+2)}n$ bins. 
 
 If $b \leq \frac{\p+2}{\p(\p+1)}$,
 we do not need small items to get this weaker upper bound.
 It is sufficient to consider the two larger intervals and use Theorem~\ref{lemma-Har-2int}, 
 since $\frac{\p^2+1}{\p(\p+1)}=\frac{\p^3+2\p^2+\p+2}{\p(\p+1)(\p+2)}$.
\end{proof}

It follows that if $(a,b)$ contains exactly two \DHk partitioning
points, then \DHk has a better competitive ratio than \DNF:

\begin{corollary}
\label{cor:abcomp}
 If $\frac{1}{\p+2}\leq a < \frac{1}{\p+1}$, then $\CRab{\DHk} > \CRab{\DNF}$.
\end{corollary}
\begin{proof}
The result follows from Theorems~\ref{lemma-DNF-3int-small} and~\ref{lemma-Har-3int}, 
since if $b \leq \frac{\p+2}{\p(\p+1)}$, then
\[
\begin{array}{rcl}
\CRab{\DHk} & = & \displaystyle \frac{\p^3+2\p^2+\p+2}{\p(\p+1)(\p+2)} = 
                  \frac{\p+1}{\p+2} + \frac{2}{\p(\p+1)(\p+2)} \\[2ex]
             & > & \displaystyle\frac{\p+1}{\p+2} \geq \CRab{\DNF}
\end{array}
\]
and otherwise,
\[
\begin{array}{rcl}
\CRab{\DHk} & = & \displaystyle \frac{\p^3+2\p^2+2}{\p(\p+1)(\p+2)} \\[2ex]
            & = & \displaystyle 
                  \frac{\p(\p+1)}{\p^2+2\p+2} +
                  \frac{2\p^3+4\p^2+4\p+4}{(\p^2+2\p+2)\p(\p+1)(\p+2)} \\[2ex]
            & > &  \displaystyle\frac{\p(\p+1)}{\p^2+2\p+2} \geq  \CRab{\DNF}
\end{array}
\]
\end{proof}

\section{Relative Worst Order Analysis}
Relative worst order analysis was introduced by Boyar and
Favrholdt~\cite{BF07j} and it
compares the performance of two algorithms $\ALG$ and $\ALGB$ directly
instead of via 
the comparison to \OPT. Algorithms are compared on the same input
sequence $I$, but on the worst possible permutation of $I$
for each algorithm.

Formally,
if $n$ is the length of $I$, and $\sigma$ is a permutation on $n$ elements,
then $\sigma(I)$ denotes $I$ permuted by $\sigma$, and
we define $\ALG_W(I) = \min_\sigma \ALG(\sigma(I))$.
If there exists a fixed constant $b$ such that, for any input sequence
$I$, $\ALG_W(I) \geq \ALGB_W(I) - b$, then $\ALG$ and $\ALGB$ are {\em comparable}
and the {\em relative worst order ratio} of $\ALG$ to $\ALGB$ is defined as follows:
$$\RWOR{\ALG}{\ALGB} = 
 \sup \SETOF{c}{\exists b \, \forall I \colon \ALG_W(I) \geq c \ALGB_W(I) - b}$$
Note that since the performance of \DHk does not depend on the
order in which the items are given, relative worst order analysis of
\DNF versus \DHk gives the same result as simply comparing the two
algorithms on each sequence separately, just as competitive analysis
with \OPT replaced by \DHk.

In~\cite{EFK12}, a relative worst order analysis of \DHk and \DNF is given
 for the model that allows items of size~1,
showing that for $i<j$, $\RWOR{\DHj}{\DHi} = \frac{i+1}{i}$.
Hence, in this model, $\RWOR{\DHk}{\DNF}=2$, for $k \geq 2$, since
 \DNF and \DHone are equivalent.
Note that, for $i \geq 2$, the result from~\cite{EFK12} holds for our
 model too, since the lower bound sequences for these cases do not
 contain items of size 1.

We first show that \DHk and \DNF are comparable.
This is a special case of the corresponding result in~\cite{EFK12}.
\begin{lemma}\label{lemma_Harbetter}
For any $k\geq 1$ and any input sequence $I$,
\[\DHk_W(I) \geq \DNF_W(I)-(k-1)\]
\end{lemma}
\begin{proof}
For any sequence $I$, we can construct an input sequence for \DNF by giving the items in the order they are packed in the bins by \DHk; first the covered bins and afterwards the items within the uncovered bins. For the closed bins, \DNF then does the same as \DHk. \DNF can cover at most $k-1$ additional bins, because \DHk has at most $k$ open bins at the end.
Thus, for any $I$, if $\sigma_{\DHk}(I)$ and $\sigma_{\DNF}(I)$ denote
 the worst permutations of $I$ with respect to the two algorithms, then
 $\DHk_W(I) = \DHk(\sigma_{\DHk}(I)) \geq \DNF(\sigma_{\DHk}(I))-(k-1) \geq
 \DNF(\sigma_{\DNF}(I))-(k-1) = \DNF_W(I)-(k-1)$. 
\end{proof}

Thus, according to relative worst order analysis, \DHk is at least as
good as \DNF.
The next lemma establishes
 a separation between the two algorithms in our model.
\begin{lemma}\label{lemma_32Harmonic}
For any $k\geq 2$, $\RWOR{\DHk}{\DNF} \geq \frac32$.
\end{lemma}
\begin{proof}
It follows from Lemma~\ref{lemma_Harbetter} that the algorithms are
comparable.

We prove that the ratio cannot be smaller than $\frac{3}{2}$ by
exhibiting a family of sequences $\SET{I_n}$ such that the following two conditions hold:
\begin{itemize}
\item
$\lim_{n\rightarrow\infty}\DHk(I_n)=\infty.$
\item
For all $I_n$, $\DHk{}_W(I_n)\geq\frac{3}{2}\cdot \DNF_W(I_n) - 1$.
\end{itemize}

For each $n \geq 1$, we define $I_n =
 \SEQ{\frac{1}{2},\SEQ{\frac{1}{2n}}^{n-1},\frac{1}{2}}^{2n}$.
\DHk covers $2n + (n-1) = 3n-1$ bins, whereas \DNF covers only $2n$
 bins.
Thus, for all $I_n$,
$${\DHk}_W(I_n)\geq\frac{3}{2}\cdot {\DNF}_W(I_n)-1\,.$$
\end{proof}

By providing a matching upper bound, we determine the exact
relative worst order ratio of the two algorithms.
\begin{theorem}
\label{theorem-RWOR}
$\RWOR{\DHk}{\DNF} = \frac{3}{2}$.
\end{theorem}
\begin{proof}
Lemma~\ref{lemma_32Harmonic} shows that
$\RWOR{\DHk}{\DNF} \geq \frac{3}{2}$.
Thus, it remains to be established that 
$\RWOR{\DHk}{\DNF} \leq \frac{3}{2}$.

Assume that an input sequence $I$ has a total volume of $n$,
and assume that \DNF covers $xn$ bins.

\textbf{Case $x < \frac{1}{2}$:} To cover fewer than $\frac{n}{2}$ bins, a
volume of more than $\frac{n}{2}$ has to be wasted by overpacking fewer than
$\frac{n}{2}$ bins. Thus, some item of size larger than one must exist,
which is a contradiction.

\textbf{Case $\frac{1}{2}\leq x <\frac{2}{3}$:} If \DNF covers only $xn$
bins, it wastes a volume of $(1-x)n$ by overpacking at most $xn$ bins.
Therefore, the average size of an item that is packed as the last
item in a bin by \DNF
is at least $\frac{(1-x)n}{xn} > \frac{1}{2}$.
Since items larger than $\frac{1}{2}$ are packed with another item
of size at least $\frac{1}{2}$ by \DHk,
the volume above $\frac{1}{2}$ is also wasted for \DHk.
Thus, \DHk wastes at least a volume of
$(\frac{(1-x)n}{xn}-\frac{1}{2})xn = n - \frac{3}{2}xn$.
So,
$\DHk(I) \leq n - (n - \frac{3}{2}xn) = \frac{3}{2} xn = \frac{3}{2} \DNF(I)$.

\textbf{Case $\frac{2}{3}\leq x \leq 1$:} The performance of \DHk is
bounded by the volume $n$ of the sequence $I$, so $\DHk(I) \leq n$.
Thus, $\DHk(I) \leq n = \frac{3}{2} \cdot \frac{2}{3}n \leq \frac{3}{2}
xn = \frac{3}{2} \DNF(I)$.
\end{proof}

We conclude that according to relative worst
order analysis, \DHk is a better algorithm than \DNF.

\section{Random Order Analysis}
The random order ratio was introduced by Kenyon~\cite{K96} as the worst ratio obtained over all sequences $I$, comparing the expected value of an algorithm~\ALG, with respect to a uniform distribution of all permutations, $\sigma$, of $I$, to the value of \OPT on $I$:
\[
 \RO{\ALG} = \liminf\limits_{\OPT(I)\rightarrow \infty} \frac{E_{\sigma}[\ALG(\sigma(I))]}{\OPT(I)}
\]
Note that \OPT is still assumed to know the entire sequence in advance,
so there is no expectation involved in computing $\OPT(I)$.

The following theorem gives a bound on how well \DNF can perform with
respect to the random order ratio.
\begin{theorem}
The random order ratio of \DNF is at most $\frac{4}{5}$.
\end{theorem}
\begin{proof}
Let $S^n$ denote all sequences of length~$n$ with item sizes from ${\cal I}$,
where ${\cal I}=\SET{\varepsilon, 1-\varepsilon}$ for an $\varepsilon$
such that
$0 < \varepsilon < \frac{1}{n}$.
Define
\[S_i^n=\SETOF{I\in S^n}{\mbox{$I$ contains $i$ items of size $\varepsilon$ and $n-i$ items of size $1-\varepsilon$}}\,.\]
Then we can consider the following disjoint partitioning 
$S^n=\bigcup_{0\leq i\leq n}S_i^n$.
We let $R^n$ denote the set of all sequences of length $n$.

The first inequality below follows from two facts:
\begin{itemize}
\item For any pair of sequences, $I, I' \in S^n_i$, $\OPT(I) = \OPT(I')$.
\item For two sums $A = \sum_{i=1}^n a_i$ and $B = \sum_{i=1}^n b_i$,
  $\frac{A}{B} \geq \min_{1 \leq i \leq n} \frac{a_i}{b_i}$.
\end{itemize}
\begin{align*}
        \frac{\EXPDIST{I \in S^n}{\DNF(I)}}{\EXPDIST{I \in S^n}{\OPT(I)}}
 & \geq \min_{0 \leq i \leq n} 
        \frac{\EXPDIST{I \in S^n_i}{\DNF(I)}}{\OPT(I^n_i)}, 
        \text{ where } I^n_i \in S^n_i\\
 & =    \min_{I \in S^n}
        \frac{\EXPDIST{\sigma}{\DNF(\sigma(I))}}{\OPT(I)} 
 \geq \min_{I \in R^n} 
        \frac{\EXPDIST{\sigma}{\DNF(\sigma(I))}}{\OPT(I)}
\end{align*}
Hence,
\[\lim_{n\rightarrow\infty}\frac{\EXPDIST{I\in S^n}{\DNF(I)}}{\EXPDIST{I\in S^n}{\OPT(I)}}
  \geq
  \liminf_{\OPT(I)\rightarrow\infty}\frac{\EXPDIST{\sigma}{\DNF(\sigma(I))}}{\OPT(I)}
  =
  \RO{\DNF}.\]
In the rest of the proof, we bound the leftmost expression from the above,
which then gives us an upper bound on the random order ratio of \DNF.

There is no difference between choosing some element from $S^n$ uniformly
at random and generating a length~$n$ sequence iteratively by choosing
the next item from ${\cal I}$ with equal probability.
Thus, we can analyze the behavior of \DNF by considering a Markov chain, where the state of the system after $i$ items have been processed is determined by the state of the open bin.
The Markov chain is finite and has just three states: either there is no open bin (N~--~for~``No''), one open bin containing one large item of size $1-\varepsilon$ (L~--~for~``Large''), or one bin with a number of small items, each of size $\varepsilon$ (S~--~for~``Small'').
Note that since $\varepsilon<\frac{1}{n}$, there is room for all the small items in one bin, if necessary.

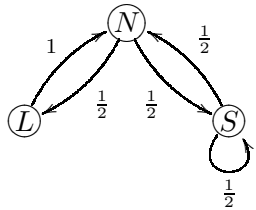
\begin{figure}[htp]
 \begin{displaymath}
 \xymatrix{ & *+[o][F-]{N} \ar@/^/[dl]^{\frac{1}{2}} \ar@/_/[dr]_{\frac{1}{2}}& \\ *+[o][F-]{L} \ar@/^/[ur]^{1} & & *+[o][F-]{S} \ar@/_/[ul]_{\frac{1}{2}} \ar@(dl,dr)_{\frac{1}{2}}}
\end{displaymath}
\caption{A Markov chain describing \DNF's behavior on the considered sequences.}
\end{figure}

This is an irreducible chain, where all states are positive recurrent,
which implies that 
it has a stationary (equilibrium) distribution, and
the probability of ending up in each of the states converges
independently of the starting state~\cite{D91b}.
The probability of being in one of the states $N$, $L$, or $S$ can be calculated from the following equations:
\begin{align*}
  1 & = \PROB{N} + \PROB{L} + \PROB{S} \\
 \PROB{N} &= \PROB{L} + \frac{\PROB{S}}{2} \\
 \PROB{L} &= \frac{\PROB{N}}{2} \\
 \PROB{S} &= \frac{\PROB{N}}{2} + \frac{\PROB{S}}{2}
\end{align*}
This system has the solution $\PROB{N}=\PROB{S}=\frac{2}{5}$ and
$\PROB{L}=\frac{1}{5}$.
From this it follows that $\EXPDIST{I\in S^n}{\DNF(I)}$ tends to $\PROB{N}n = \frac{2}{5}n$.

For the optimal algorithm, note that its result only depends on the
number of items of each size. In particular, after $n$ items, it can
cover $\FLOOR{\frac{n}{2}}$ bins, unless there
are more small than large items.
All the extra small items would be wasted.

Using random walks, it is easy to see that the expected difference between
the number of large and small items is a low order term compared with $n$,
and therefore does not affect the limit.
A sequence of independent stochastic variables $\SET{X_i}_{i\geq1}$,
where $\PROB{X_i=1}=\PROB{X_i=-1}=\frac{1}{2}$,
is called a {\em simple random walk}~\cite{D91b}.
It is well known that if we define $T_n=\sum_{i=1}^n X_i$,
then $\lim_{n\rightarrow\infty}\frac{\EXP{|T_n|}}{\sqrt{n}}=\sqrt{\frac{2}{\pi}}$~\cite{H94b}.
Hence, $\EXP{|T_n|}\in O(\sqrt{n})$,
and then
$\EXPDIST{I\in S^n}{\OPT(I)} = \frac{n}{2} - O(\sqrt{n})$.

In conclusion, we get
\[
 \lim\limits_{n \rightarrow \infty}
    \frac{\EXPDIST{I\in S^n}{\DNF(I)}}{\EXPDIST{I\in S^n}{\OPT(I)}}
  = \lim\limits_{n \rightarrow \infty}\frac{\frac{2}{5}n}{\frac{n}{2}-O(\sqrt{n})}
  = \frac{4}{5}.
\]
\end{proof}

\begin{theorem}
The random order ratio of \DHk is $\frac{1}{2}$.
\end{theorem}

\begin{proof}
The performance of \DHk does not depend on the order of the items in the sequence. Given a sequence containing $n$ items of size $1-\varepsilon$ and $n$ items of size $\varepsilon$, where $\varepsilon < \frac{1}{n}$, \DHk will always cover $\frac{n}{2}$ bins, while \OPT will cover $n$ bins.
The lower bound is given by Theorem~\ref{competitive-reasonable},
since the random order ratio of a bin covering algorithm is never
worse than its competitive ratio.
\end{proof}

Thus, according to random order analysis, \DNF is at least as good as \DHk.
Though it seems hard to raise the lower bound on the random order
ratio for \DNF above $\frac{1}{2}$, and thereby separate the two
algorithms, we conjecture that \DNF is in fact strictly better than \DHk
with respect to this measure. We discuss this further in the conclusion.

\section{The Max/Max Ratio}
The max/max ratio was introduced by Ben-David and Borodin~\cite{BB94}
and compares an algorithm's worst-case behavior on any sequence of length $n$
with \OPT's worst-case behavior on any sequence of length $n$.

The max/max ratio was introduced for the minimization problems paging
and $K$-server. Since bin covering is a maximization problem,
we actually need a min/min ratio.
Additionally, since the input items can be arbitrarily small, letting
the sequence length approach infinity does not give interesting results.
Thus, we modify the measure to consider the volume, \vol{I}, of a
sequence $I$, where \vol{I} is the sum of the sizes of all the items in $I$:
\begin{align*}
\MINV{\ALG} = \dfrac{\liminf_{v\rightarrow\infty}\min_{\vol{I} = v} \ALG(I)/v}{\liminf_{v\rightarrow\infty}\min_{\vol{I} = v} \OPT(I)/v}
\end{align*}

It turns out that this measure cannot distinguish between \DNF and \DHk
in the general case:

\begin{theorem}
\label{max-not-separate}
Both \DNF and \DHk have a min/min ratio of $1$.
\end{theorem}
\begin{proof}
For any $\varepsilon >0$, a sequence consisting only of items of size
$1-\varepsilon$ will force any algorithm, including \OPT, to put at
least two items in each bin. As $\varepsilon$ tends to 0, this gives
an upper bound on the number of covered bins tending to $\vol{I}/2$.
Since both \DNF and \DHk always cover at least $\lfloor \vol{I}/2
\rfloor$ bins, this shows that their min/min ratios are $1$.
%
\end{proof}

If the item sizes are restricted to an interval 
$(a,b)\subseteq (0,1)$ containing at least one \DHk interval border,
the min/min ratio can distinguish between \DNF and \DHk.
If $(a,b)$ does not contain at least one of the interval borders
used by \DHk, then \DHk packs exactly like \DNF.

If $(a,b)$ contains a \DHk border, then we define, as in Section~\ref{competitive-analysis}, 
$\frac{1}{\p}$ as the {\em maximal border in $(a,b)$}.
Throughout the paper, we assume that the constants $k$, $a$, $b$, and $p$
have the meaning defined above.

\begin{theorem}
\label{theorem-max-max-DHk}
With item sizes in $(a,b)\subseteq (0,1)$, where $a < \frac{1}{\p}$,
\DHk has a min/min ratio of $1$.
\end{theorem}
\begin{proof}
The worst-case sequences for \DHk consist of items only of size either $b-\varepsilon$ or $\frac{1}{\p}-\varepsilon$,
for any small $\varepsilon$, and, since there are no choices in packing
sequences with just one item size, \OPT cannot pack them better than \DHk.
\end{proof}

\begin{theorem}
\label{max-dnf}
With item sizes in $(a,b)\subseteq (0,1)$, where $\frac{1}{\p}\in (a,b)$,
\DNF has a min/min ratio of
$\max\SET{\frac{1+\frac{1}{\p}}{1+b},\frac{\p b}{1+b}}$.
\end{theorem}
\begin{proof}
To maximize the overpacking by \DNF, the last item of each bin
should have size close to $b$ and be packed in a nearly full bin.
Thus, we arrange that each bin gets $\p$ items of size $\frac{1}{\p}-\varepsilon$
for some $0 < \varepsilon < \frac{1}{\p}-a$, and then an item of size
$b-\eps$.
Each bin receives a volume of $1-\p\varepsilon+b-\eps$, so to use volume~$n$,
we repeat this $n/(1-(\p+1)\varepsilon+b)$ times to get a sequence $I_n$.
We may assume this is integral, since any rounding disappears in the limit,
$$\liminf\limits_{n\rightarrow\infty}\min\limits_{\vol{I} = n} \frac{\DNF(I_n)}{n}
=\frac{1}{1-(\p+1)\varepsilon+b}\,,$$ and, since we can use any
$\varepsilon$, $0 < \varepsilon < \frac{1}{\p}-a$,
we arrive at $\frac{1}{1+b}$.

The worst-case for \OPT follows by using one of the two types of sequences
from the proof of Theorem~\ref{theorem-max-max-DHk},
i.e., for each bin, $\p +1$ items of size $\frac{1}{\p}-\varepsilon$
or $\p$ items of size $b-\varepsilon$, for some $\varepsilon$.
Similar to the calculations above, and letting $\varepsilon$ approach zero,
the limit for \OPT becomes
$$\liminf\limits_{n\rightarrow\infty}\min\limits_{\vol{I} = n} \frac{\OPT(I_n)}{n} = \min\SET{\frac{1}{(\p+1)\frac{1}{\p}},\frac{1}{\p b}}
=\min\SET{\frac{1}{1+\frac{1}{\p}},\frac{1}{\p b}}\,.$$

Dividing the result for \DNF with the result for \OPT,
we get the stated ratio.
\end{proof}

Note that $\frac{1+\frac{1}{\p}}{1+b}<1$ is equivalent to
$\frac{1}{\p}<b$, which follows from the definition and maximality of
$\frac{1}{\p}$.
Furthermore, $\frac{\p b}{1+b}<1$ is equivalent to $b<\frac{1}{\p -1}$,
which is satisfied as long as $b$ is not equal to $\frac{1}{\p -1}$.
Thus, according to min/min analysis, \DHk is better than \DNF when
item sizes are restricted to an interval $(a,b) \in (0,1)$ containing
at least one \DHk border, and $b\not=\frac{1}{\p -1}$ where $\frac{1}{\p}$ is the
maximal border.

\section{Uniform Distribution}

In this section, we study the expected performance ratio of \DNF and \DHk on sequences
containing items drawn uniformly at random from the interval~$(0,1)$.

The expected performance ratio $\ERU{\ALG}$ is the ratio between the expected performance of 
the algorithms \ALG and \OPT on sequences of length $n$, containing items drawn 
uniformly at random from the interval $(0,1)$:

$$\ERU{\ALG} = \lim_{n \rightarrow \infty}
  \frac{\EXPDIST{I\in U_n(0,1)}{\ALG(I)}}{\EXPDIST{I\in U_n(0,1)}{\OPT(I)}}.$$

\begin{theorem}
\label{theorem-uniform}
 On a sequence containing items drawn uniformly at random from the
 interval $(0,1)$, 
 \begin{align*}
 & \ERU{\DHtwo} = \frac{1}{2}+\frac{1}{e^2-e}\approx 0.7141 \text{ and }\\
 & \lim_{k \rightarrow \infty} \ERU{\DHk} = \frac{12-\pi^2}{3} \approx 0.7101\,.
 \end{align*}

\end{theorem}
\begin{proof}
For sequences of length $n$ drawn uniformly at random from the interval $(0,1)$,
 \begin{align*}
 \ERU{\DHk}
& = \lim_{n \rightarrow \infty}
  \frac{\EXPDIST{I\in U_n(0,1)}{\DHk(I)}}{\EXPDIST{I\in U_n(0,1)}{\OPT(I)}} \\
& = \lim_{n \rightarrow \infty}
  \frac{\EXPDIST{I\in U_{\frac{(k-1)n}{k}}[\frac{1}{k},1)}{\DHk(I)}
        +\EXPDIST{I\in U_{\frac{n}{k}}(0,\frac{1}{k})}{\DHk(I)}}
       {\EXPDIST{I\in U_n(0,1)}{\OPT(I)}} \\
& = \RHar + \RDNF \,,
 \end{align*}
where the second equality follows from the fact that \DHk processes
items smaller
than $\frac{1}{k}$ separately from items of size at least
$\frac{1}{k}$. Thus, these items can be treated separately.
Since item sizes are chosen uniformly at random and
the result of \DHk depends linearly on the number of items in each interval,
this corresponds to scaling $n$ using $\frac{k-1}{k}$ and $\frac{1}{k}$,
respectively. 

The final equality just defines the following two expressions as
 $$\RHar = \lim_{n \rightarrow \infty}
  \frac{\EXPDIST{I\in U_{\frac{(k-1)n}{k}}[\frac{1}{k},1)}{\DHk(I)}}
       {\EXPDIST{I\in U_n(0,1)}{\OPT(I)}}$$
and
$$   \RDNF = \lim_{n \rightarrow \infty}
  \frac{\EXPDIST{I\in U_{\frac{n}{k}}(0,\frac{1}{k})}{\DHk(I)}}
       {\EXPDIST{I\in U_n(0,1)}{\OPT(I)}}\,.$$

Using a pairing heuristic, \cite{CFGK91} shows that  
 $\EXPDIST{I\in U_n[0,1)}{\OPT(I)} = \frac{n}{2}$.

For a sequence with items drawn uniformly at random from $(0,1)$,
 the expected number of items with sizes
 in the interval $[\frac{1}{i},\frac{1}{i-1})$ is
 $(\frac{1}{i-1}-\frac{1}{i})n = \frac{n}{i(i-1)}$.
For $2 \leq
 i \leq k$, \DHk packs each of these items (except for at most $i-1$ items) 
 in bins with exactly $i$ items each.
Hence, the expected number of bins that \DHk
 covers with such items is more than
 $\frac{n}{i^2(i-1)}-1$, $2 \leq i \leq k$.

Hence,
 \begin{align*}
   \RHar
 & = \lim_{n \rightarrow \infty} 
     \frac{\sum\limits_{i=2}^k \left(\frac{n}{i^2(i-1)}-1\right)}
          {n/2}
   = 2 \sum\limits_{i=2}^k \frac{1}{i^2(i-1)}\,.
 \end{align*}
Using partial fraction decomposition, we get
 \begin{align*}
   \RHar
 & = 2 \sum\limits_{i=2}^k  \left(\frac{1}{i-1} - \frac{1}{i} -
     \frac{1}{i^2}\right)\\ 
 & = 2 \left( \sum\limits_{i=1}^{k-1}\frac{1}{i} -  
     \sum\limits_{i=2}^k \frac{1}{i} -
     \sum\limits_{i=1}^k \frac{1}{i^2} + 1 \right) \\
 & = 2 \left( 1 - \frac{1}{k} - 
     \sum\limits_{i=0}^{k-1}\frac{1}{(i+1)^2} + 1 \right) \\
 & = 2 \left( 2 - \frac{1}{k} - 
     \sum\limits_{i=0}^\infty \frac{1}{(i+1)^2} +
     \sum\limits_{i=0}^\infty \frac{1}{(i+1+k)^2} \right)\\
 & = 2 \left( 2 - \frac{1}{k} - \gf(1) + \gf(k+1) \right)\,,
 \end{align*}
 where \gf is the trigamma function~\cite{AS64b}.
Some properties of \gf are that $\gf(1)=\frac{\pi^2}{6}$, $\gf(k+1)
 = \gf(k)- \frac{1}{k^2}$, and $\gf(k) \rightarrow 0$ as $k
 \rightarrow \infty$.
Now,
 \begin{align*}
   \RHar
 & = 2\left(2 - \frac{1}{k} - \frac{\pi^2}{6} + \gf(k)-\frac{1}{k^2}\right)\\
 & = 2\left(\frac{12-\pi^2}{6} - \frac{1+k}{k^2} + \gf(k)\right).
 \end{align*}

Since $\RDNF \rightarrow 0$ as $k \rightarrow \infty$, it follows that 
 \begin{align*}
     \lim_{k \rightarrow \infty} \ERU{\DHk}
 & = \frac{12-\pi^2}{3} \approx 0.7101\,.
 \end{align*}

Since \DHk packs the items of sizes in $(0,\frac{1}{k})$ the same
way \DNF would,
we can use a result from~\cite{CFGK91}, stating that
\[
\lim\limits_{n\rightarrow\infty}
  \frac{\EXPDIST{I\in U_n[0,\frac{1}{k})}{\DNF(I)}}{n}
=
\frac{1}{\mu(k)}
\]
where
 $$\mu(k) 
 = \lim\limits_{\varepsilon \rightarrow 0}
   \sum\limits_{l=0}^k (-1)^l \frac{1}{l!}
      \left( \frac{k}{1-\varepsilon} - l \right)^l
      e^{\frac{k}{1-\varepsilon} -l}
 = \sum\limits_{l=1}^{k}  \frac{e^l (-l)^{k-l}}{(k-l)!}\,,$$
and the limit for $\varepsilon\rightarrow 0$ is due to our working with
an open interval, where the interval in~\cite[Eq.~(28)]{CFGK91} is closed.
Now,
 \begin{align*}
   \RDNF 
 & = \frac{1}{k}
     \lim\limits_{n\rightarrow\infty}
     \frac{\EXPDIST{I\in U_n[0,\frac{1}{k})}{\DNF(I)}}{n/2}
   = \frac{2}{\mu(k) k}\,.
 \end{align*}

Note that $\mu(2) = e^2-e$ 
and $\gf(2) = \gf(1) - \frac{1}{1^2} = \frac{\pi^2}{6} - 1$.

Hence, 
 \begin{align*}
   \ERU{\DHtwo} 
 & = \RHartwo + \RDNFtwo \\
 & = 2 \left( \frac{12-\pi^2}{6} - \frac{1+2}{2^2} + \gf(2) \right) + 
     \frac{2}{2\mu(2)}\\
 & = \frac{12-\pi^2}{3} - \frac{3}{2} + \frac{\pi^2}{3} - 2 + 
     \frac{1}{e^2-e} \\
 & = \frac{1}{2} + \frac{1}{e^2-e} \approx 0.7141.
 \end{align*}
\end{proof}

This should be compared with a result from~\cite{CFGK91}, showing
that on a uniform distribution,
\DNF has an expected performance ratio of $\frac{2}{e}\approx 0.7358$.
Thus, under this assumption, \DNF is a little better than \DHk.

\section{Concluding Remarks}
The starting point for this paper was the fact that bin covering
algorithms as different as \DNF and \DHk are not separated using
competitive analysis.
We are interested in the question of which algorithm to use in different
scenarios. \DHk was designed to guard against worst-case sequences,
and since these are often made up using pathological input, such as mixing very
large and very small items, we have carried out analyses using the
worst-case performance, but on restricted input of items of similar size.
The comparison is still in \DHk's favor, albeit less so.
Max/max analysis (under similar conditions) and relative worst order
analysis also point to \DHk.

In contrast, if input is not organized into worst-case sequences by
an adversary, we can show, by carrying out an analysis of the expected
results under a uniform distribution that \DNF performs a little
better than \DHk. This seems to be very robust, since adding a small
element of worst-case requirements in the form of random order analysis
also points to \DNF not being worse than \DHk. Thus, even if an adversary
gets to choose the worst sequence for the algorithm, just the fact that
the items are received in a random order removes
\DHk's advantage over \DNF.

The conclusion is that unless guarantees are desired or it is known that
items do not arrive in a random order, it is worth considering \DNF
as the algorithm of choice.

\DHk has a random order ratio of $\frac{1}{2}$,
which is worst possible, whereas the upper bound we have on
\DNF is $\frac{4}{5}$. We conjecture that these two algorithms
can be separated, proving \DNF to be best.
Though it is not essential to the conclusion above,
we leave this as an interesting open problem we would like to see solved,
and use the rest of this section to discuss some relevant issues
regarding this.
It seems intuitively almost obvious that
\DNF would always get a ratio larger than $\frac{1}{2}$.
The difficulty in establishing this formally stems from problems of
handling the size aspects using probability theory. In the hardest
case, there are a linear number of very large items such that if
they end up on top of each other pairwise, we get the ratio of $\frac{1}{2}$. 
Thus, we need to prove that some fraction of these large items do not end up pairwise
on top of each other.
The small items that would be packed with the large items in an
optimal packing can be cut into very small pieces so there are
orders of magnitude more small items than large items---but still of
possibly dramatically varying size, relatively.
Whereas we have strong theoretical tools for
bounding the deviation from the expected
number of items in certain locations in the form of
Chebyshev's inequality, for instance, it is much harder to reason
regarding deviations from the expected size, and it is exactly the sum
of sizes of small items surrounding a large item that decides whether or not
two large items end up on top of each other.

Results on the random order ratio are often difficult to establish.
This is reflected in the rather small
number of obtained results and also in published
results being far from tight. In the paper~\cite{K96}
introducing the random order ratio, for example,
the random order ratio of the bin
packing algorithm Best-Fit is shown to lie between 1.08 and 1.5.
An exceptionally tight result appears in~\cite{CCRZ08}, where it is shown
that the random order ratio of Next-Fit for bin packing is exactly~$2$.
Note, however, that this result does not give indication that the
random order ratio of \DNF for bin covering should be
$\frac{1}{2}$. The sequence establishing the lower bound of 2 consists of
$n$ items of size~$\frac{1}{2}$ and $kn$ items of size~$\epsilon <
\frac{1}{kn}$, for some large $k$. For a random ordering of these
items, each item of size $\frac{1}{2}$ has a high probability of
being combined with at least one of the small items, leaving too
little space in the bin for another large item. 
For bin covering, the problem is reversed;
to prove an upper bound of $\frac12$,
we must prove that each large item has a
significant probability of being surrounded by a sufficiently small
volume of small items so that it will go into the same bin as a
neighboring large item.
%

\bibliography{refs}
\bibliographystyle{plain}

\end{document}